\documentclass[%
 reprint,
 amsmath,amssymb,
 aps,
prl,
floatfix
]{revtex4-1}
\usepackage{hyperref}
\usepackage{braket}
\usepackage{amsthm}
\usepackage{amsbsy}
\usepackage{graphicx}
\usepackage{dcolumn}
\usepackage{bm}
\usepackage{color}
\usepackage[normalem]{ulem}
\usepackage{bbm}
\usepackage{tikz,tkz-graph}
\usepackage{ifthen}
\usepackage[EULERGREEK]{sansmath}
\usepackage{nicefrac}

\usepackage[printwatermark]{xwatermark}
\usepackage{xcolor}

\newif \ifcolor
\colortrue

\newcommand{\yanote}[1]  {\textcolor{red}{{\bf (Yosi:}{#1}{\bf )}}}

\newcommand{\abs}[1]{\left| #1 \right|}

\newtheorem{myclaim}{Claim}

\newcommand{\prnt}[1]{\left( #1 \right)}
\newcommand{\prntt}[1]{\left[ #1 \right]}\newcommand{\prnttt}[1]{\left\{ #1 \right\}}
\newcommand{\ketbra}[2]{ \left| #1 \right\rangle\left\langle #2 \right|}
\long\def\/*#1*/{}

\makeatletter
\newcommand{\doublewidetilde}[1]{{%
  \mathpalette\double@widetilde{#1}%
}}
\newcommand{\double@widetilde}[2]{%
  \sbox\z@{$\m@th#1\widetilde{#2}$}%
  \ht\z@=.9\ht\z@
  \widetilde{\box\z@}%
}
\makeatother

\begin{document}
%
%

\bibliographystyle{apsrev4-1}

\preprint{APS/123-QED}

\title{Robust Diabatic Quantum Search by Landau-Zener-St\"uckelberg Oscillations}

\author{Yosi Atia$^1$}
\email[Corresponding author. Email:]{g.yosiat@gmail.com}
\author{Yonathan Oren$^1$}
\author{Nadav Katz$^2$}

\affiliation{$^1$The Rachel and Selim Benin School of Computer Science and Engineering, The Hebrew University, Jerusalem 91904,  Israel}
\affiliation{$^2$ The Racah Institute of Physics, The Hebrew University, Jerusalem 91904,  Israel
}

\date{\today}

\begin{abstract}
Quantum computation by the adiabatic theorem requires a slowly varying Hamiltonian with respect to the spectral gap.
We show that the Landau-Zener-St{\"u}ckelberg oscillation phenomenon, that naturally occurs in quantum two level systems under non-adiabatic periodic drive, can be exploited to find the ground state of an N dimensional Grover Hamiltonian. The total runtime of this method is $O(\sqrt{2^n})$ which is equal to the computational time of the Grover algorithm in the quantum circuit model. An additional periodic drive can suppress a large subset of Hamiltonian control errors using coherent destruction of tunneling, providing superior performance compared to standard algorithms. 
\end{abstract}
	
\maketitle


Adiabatic Quantum Computation (AQC)~\cite{FGGS00,AL18} is a computational model, motivated by the physical phenomenon described by the adiabatic theorem, which states that if a system is prepared in the ground state of an initial Hamiltonian, and the Hamiltonian slowly varies in time, then  it is guaranteed that the evolution will be adiabatic - meaning that the system will remain close to its instantaneous ground state throughout \cite{Kato1950,Messiah}. By encoding a solution for a computational problem in the ground state of the finally applied Hamiltonian, one can exploit this phenomenon to produce the aforementioned ground state, and thus produce a solution to the problem. The maximal rate of change allowed for such evolution usually scales with the inverse square of the energy gap between the ground state and the first excited state~\cite{FGGS00}.

The Grover problem~\cite{Grover96}, also known as \emph{The Unstructured Search Problem} is one of the few problems solvable by a native adiabatic algorithm, which achieves the same performance as the best possible algorithm in the circuit model \cite{BBBV94} (for other native algorithms see \cite{Hen14} and the partially adiabatic \cite{SNK12}). The input to the problem is an $n$ qubit Hamiltonian, which can only be used as a black box, i.e., can be switched on or off \footnote{We have used units very loosely in this work. See discussion  at the Supplementary Material [URL will be inserted by publisher]}

\begin{equation}
H_{p}=I_N-\ket{y}\bra{y},
\end{equation}
where $I_N$ is the $N\times N$ identity matrix with $N=2^n$, and the problem is to find the unknown string $y$. The problem is comparable to finding the ground state of a known multiple-qubit Hamiltonian; the ground state might be computationally hard to find and therefore can be considered ``computationally unknown'' \cite{AA17}.

An adiabatic algorithm for the search problem was suggested by \cite{FGGS00}.
The system is initialized to a symmetric superposition of states denoted $\ket{u}=\ket{+\dots+}$,  and then evolves by the time-dependent Hamiltonian 
\begin{equation} \label{eq:AdiGrover}
\begin{split}
H_{G}\left(s(t)\right)=&
 (1-s(t))\cdot (I_N-\ket{u}\bra{u})
 \\
 &+ s(t)\cdot(I_N-\ket{y}\bra{y}),
\end{split}
\end{equation}
where the \emph{control function} 
$s(t):[t_i,t_f]\rightarrow[0,1]$
is initialized to 0 and increases monotonically with time to 1.  
The minimal gap for $n$ qubit systems is $\Delta=\sqrt{2^{-n}}$. Evolving with a linear $s(t)$  requires $O(2^n)$ time, while a specially tailored control function, whose rate matches the instantaneous spectral gap, generates the ground state of $H_p$ in the optimal time, $O(\sqrt{2^{n}})$ \cite{DMV01,RC02}.

In this work, we introduce a \emph{diabatic} algorithm for the Grover problem, denoted algorithm $\mathcal A$, whose performance matches both the optimized adiabatic and the circuit model algorithms \cite{RC02, Grover96,BBBV94}, by setting $s(t)=(1-A\cos(\omega t))/2$ where $\omega\gg \Delta$. The system passes the minimal gap multiple times diabatically and is effectively evolving by a Landau-Zener-Stu\"kelberg (LZS) Hamiltonian \cite{Landau32,Zener32, Stueckelberg32,OV09,SAN10}. Abandoning adiabaticity gave us more freedom in algorithm design. In algorithm $\mathcal B$, we add an oscillating term $B\cos (\omega t) \ketbra{u}{u}$ which yields improved robustness to Hamiltonian control errors relative to previous algorithms \cite{Grover96,RC02}.

We start by analyzing the Landau-Zener-Stuckelberg Hamiltonian (for a generic  two level system with bare states $\ket{0},\ket{1}$):
\begin{equation}
\begin{split}
H_{\mathrm{LZS}}(t) :=& \frac{1}{2} \left(
-A\cos (\omega t)\sigma_z -\Delta \sigma_x  
\right)\\
=&\frac{1}{2}\begin{bmatrix}
	-A\cos(\omega t) & -\Delta \\
	-\Delta &A \cos(\omega t)
\end{bmatrix}.
\end{split}
\end{equation}
The sinusoidal drive causes the Hamiltonian to exhibit avoided level crossings at $t=\pi(k+\frac{1}{2})/\omega$ for $k\in \mathbbm{N}$ with a minimal energy gap of $\Delta$ (see Fig. \ref{fig:rate_and_energy}).

\begin{figure} [h]
\begin{tikzpicture} [scale=1]
\ifcolor
\else
\selectcolormodel{gray}
\fi
\fill [gray!30] (0,-1) rectangle (4,4.5);
\fill [green!30] (4,-1) rectangle (8,4.5);


\draw[thin,<-] (8.3,0) -- (0,0) node [left] {\rotatebox{90} {\mbox {$A\cos(\omega t)$}}};
\draw[color=purple]   plot [domain=0:8.2,samples=100] (\x,{cos(\x r /2 *pi)});

\foreach \x in {1,2,3,4,5,6,7,8}
{
\draw [very thin, color=gray]  (\x,-1) -- (\x,1) node[shift={(0,0.5)}, color=black] { {  
{\ifthenelse{\equal{\x}{0}}
{$0$}
{{\ifthenelse{\equal{\x}{1}} 
{$\frac{\pi}{2\omega}$}
{$\frac{\x\pi}{2\omega}$} 
}}}}};}

\foreach \x in {1,2,3,4,5,6,7,8}
{
\draw [very thin, color=gray]  (\x,2) -- (\x,4);
}

\draw [->, very thin]  (0,2) -- (0,4.5);
\draw [->, very thin]  (0,-1) -- (0,1.5);

\foreach \x in {0,4}
{
\node at (\x+.2, 2.3) {\mbox {$\ket{0}$}};
\node at (\x+.2, 3.7) {\mbox {$\ket{1}$}};
}

\foreach \x in {2,6}
{
\node at (\x+.2, 2.3) {\mbox {$\ket{1}$}};
\node at (\x+.2, 3.7) {\mbox {$\ket{0}$}};
}

\draw[thin,<-] (8.3,3) -- (0,3) node [left] {\rotatebox{90} {Energy}};

\node[] at (8.1,2.85) {\mbox{$t$}};
\node[] at (8.1,-0.15) {\mbox{$t$}};

\draw[color=red]   plot [domain=0:8.2,samples=300] (\x,{sqrt(cos(\x r /2 *pi)^2+0.01)+3});
\draw[color=blue]   plot [domain=0:8.2,samples=300] (\x,{-sqrt(cos(\x r /2 *pi)^2+0.01)+3});

\end{tikzpicture}
\caption{\label{fig:rate_and_energy}{\ifcolor\else(Color online) \fi}Top: the instantaneous eigenvalues of $H_{\mathrm{LZS}}(t)$; bottom: the drive $A \cos (\omega t)$. Avoided crossings occur at  $t=\pi (k+\frac{1}{2}) / \omega$ for integer $k$, when $\cos(\omega t)=0$. Each period of the drive (gray or green background) contains a double crossing. Note that the ground state and the excited state alternate at every avoided crossing.}
\end{figure}
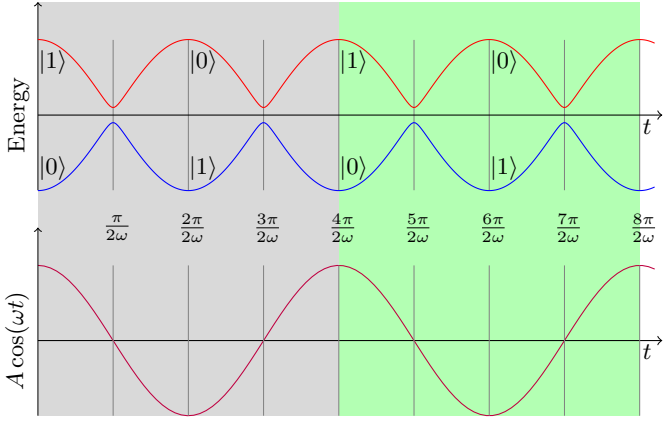

In order to gain some intuition, consider a system initialized to the state $\ket{0}$ and driven through the avoided crossing twice (i.e., one period of $s(t)$). After the double-crossing, the population of the  state $\ket{1}$, denoted $P_+^{(2)}$ approaches 0 for both $\omega\ll \Delta^2/A$ and for $\omega \gg A$: if $\omega \ll \Delta ^2/A$, the adiabatic condition holds, the system follows the ground state at all times, and thus returns to $\ket{0}$. In  the limit $\omega\gg A$, the propagator approaches unity and the state remains unperturbed. In intermediate cases an interesting phenomenon occurs: in the first passage of the avoided crossing the system transfers almost perfectly from the initial ground state to the final excited state, however a tiny amplitude leaks to orthogonal state. The populations of the excited state and the ground state gain different phases between the two  crossings, and finally interfere again in the second crossing.  $P_+^{(2)}$ is affected by this interference and oscillates with the periodicity of the control $2\pi/\omega$ in what is known as \textit{Landau-Zener-Stuckelberg oscillations}  \cite{Landau32,Zener32, Stueckelberg32} (See Fig. \ref{fig:main_plot}).

In the regime $\omega \gg \Delta$  one can use the rotating wave approximation  
(see \cite{AJZN07}, \footnote{Supplemental Material at [URL will be inserted by publisher]}) to show that with periodic drive the system oscillates around the $x$ axis in the Bloch sphere with frequency 
\begin{equation} \label{eq:Omega}
\Omega = \Delta \left|J_0 \left(\frac{A}{\omega}\right)\right|.
\end{equation}
The algorithm will fail when $A/\omega$ equals a root of the Bessel function $J_0$, where a coherent destruction of transition (CDT) occurs, and $\Omega=0$ (\cite{GDJH91}, see also \cite{AJZN07,SAN10}). CDT was previously suggested as a method  to control interactions in quantum systems \cite{VUS04, LSCSZMA07,ZGKH09} and we use these ideas in algorithm $\mathcal B$.

\begin{figure} 
\centering
\ifcolor
\includegraphics[scale=0.5] {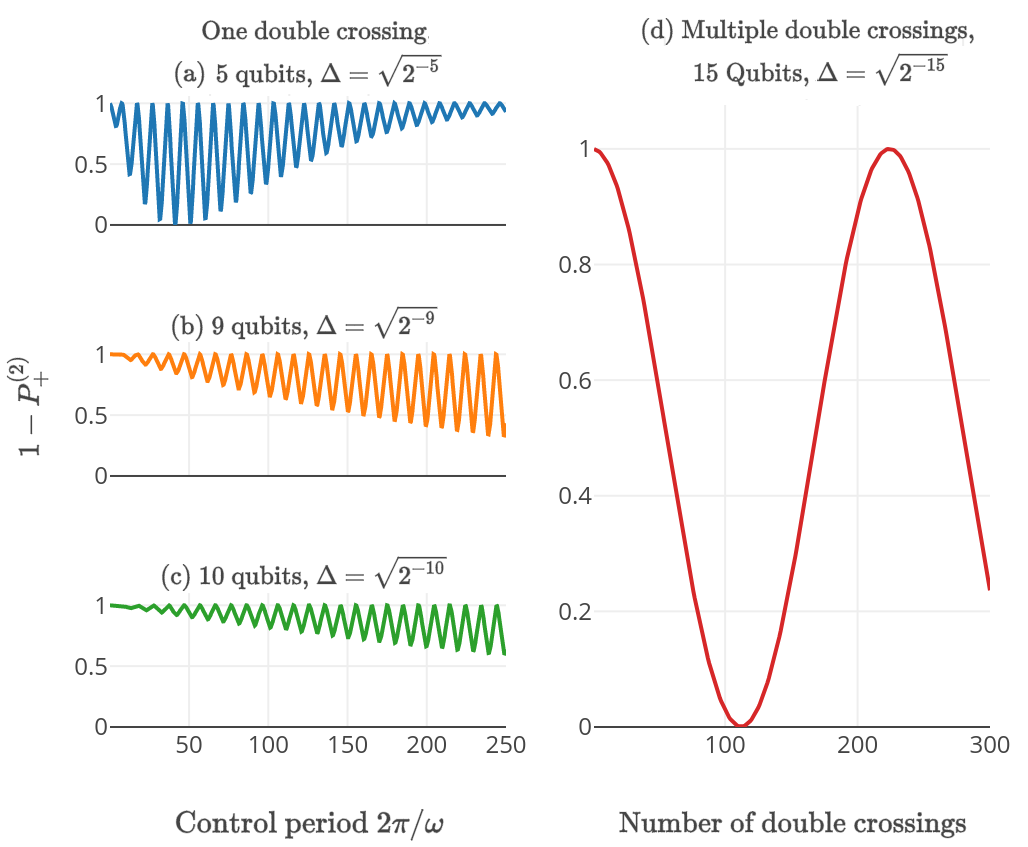}
\else
\includegraphics[scale=0.5] {P_+_2gray}
\fi
\caption{\label{fig:main_plot}
{\ifcolor\else(Color online) \fi}Numerical simulation of LZS  oscillations solving the Grover problem where the system is initialized to the ground state at $t=0$. (a)-(c) - the ground state population after a double crossing with different $\omega$ and gaps. This probability reaches 1 both for $\omega \gg A$ and for $\omega \ll \Delta^2/A$ (only visible in (a)). For the first limit the system is almost unperturbed, while in second limit the process is adiabatic and the system follows the instantaneous ground state and returns to its initial state. While the rotating wave approximation holds ($\omega \gg \Delta$),  the system oscillates by the Rabi frequency $\Omega=\Delta \left|J_0 \left(\frac{A}{\omega}\right)\right| \cdot 2\pi/\omega$. The zeros of the Bessel function correspond to coherent destruction of tunneling, where $1-P_+^{(2)}=1$ in the graph.  The approximation fails as $\omega\lesssim \Delta$ in (a). (d) Numerical simulation of the ground state population following multiple double crossings in a 15-qubit system.}  
\end{figure}

Interestingly, the Grover Hamiltonian $H_\mathrm{G}(t)$ with a periodic control function is closely related to $H_{\mathrm{LZS}}(t)$. The key to the mapping is the invariance of the subspace $V=\mathrm{span} \left\{\ket{u},\ket{y} \right\}$ to $H_\mathrm{G}(s)$ for all $s$ \footnote{See Supplemental Material at [URL will be inserted by publisher] for proof.}. Although  $V$ is isomorphic to the Hilbert space of a 2-level system, one cannot map $\ket{u},\ket{y}$ to $\ket{0},\ket{1}$ trivially in $H_\mathrm{LZS}$ since the first pair is only approximately orthogonal. To overcome this problem  we define a new basis $\ket{\bar 0},\ket{\bar 1}$, exponentially close to $\ket{u}$ and $\ket{y}$, as stated in the following claim \footnote{See Supplemental Material at [URL will be inserted by publisher] for proof.}:

\begin{myclaim} \label{clm:equiv}
The projection of $H_{\mathrm G}(s(t))$ on $V$ satisfies:

\begin{equation} \label{eq:HinV}
H_{\mathrm G}(s(t)) \bigg|_V = \frac{I_2}{2} + \left(s(t)-\frac{1}{2}\right) \sqrt{1-\Delta^2} \bar\sigma_z -\frac{\Delta}{2}\bar\sigma_x, 
\end{equation}
where $\Delta=\braket{y|u}$. The operators $\bar \sigma_x ,\bar \sigma_z$ act on the states
\begin{equation}
\begin{split}
\ket{\bar 0}&=\sqrt{\frac{1+\sqrt{1-\Delta^{2}}}{2}} \ket{u} + \sqrt{\frac{1-\sqrt{1-\Delta^{2}}}{2}}\ket{u^\perp} \\
\ket{\bar 1}&=\sqrt{\frac{1-\sqrt{1-\Delta^{2}}}{2}} \ket{u} - \sqrt{\frac{1+\sqrt{1-\Delta^{2}}}{2}}\ket{u^\perp},
\end{split}
\end{equation}
where  ${\ket{u^\perp}\mathrel{\mathop:}=
 \frac{\ket{y}-\Delta\ket{u}} {\sqrt{1-\Delta^2}}}$ is the vector orthogonal to $\ket{u}$ in $V$.
\end{myclaim}

Algorithm $\mathcal A$ is an immediate corollary of Claim \ref{clm:equiv}.
The Hamiltonian $H_\mathrm{G}$ with a control function $s(t)=(1 -A\cos(\omega t))/2$ acts on $V$ as an LZS Hamiltonian on the states $\ket{\bar 0},\ket{\bar 1}$. Since $\ket {\bar 0}$ and $\ket{\bar 1}$ are exponentially close to $\ket{u}$ and $\ket{y}$ respectively,   evolving $\ket{u}$ by $H_\mathrm{G}(s(t))$ will cause the system to oscillate between the states close to $\ket{u}$ and $\ket{y}$ with frequency $\Omega=\Delta \abs{J_0(\sqrt{1-\Delta^2}A/\omega)}$. Hence, such a driven Hamiltonian can solve the Grover problem in time $O(\sqrt{2^n})$ - the same complexity as the optimized circuit and adiabatic models.

A careful analysis of LZS interferometry shows that the algorithm finds $y$ for a wide range of  $A,\omega$. We require only $\omega \gg \Delta$ for the rotating wave approximation to hold. $J_0(\sqrt{1-\Delta^2}A/\omega)$ is a factor of the algorithm's run-time, hence $A/\omega$ should not be large (for $z\gg 1$, $J_0(z) \approx 1/ \sqrt{z}$ ), and not too close to the roots of $J_0$ as it will cause $\Omega$ to diminish by CDT. Note that none of these constraints requires a prior knowledge of the gap $\Delta$, other than an upper bound, hence the algorithm is robust to an  multiplicative error of the Hamiltonian due to calibration errors. 

The limit $A=0$ yields maximal $\Omega$, and corresponds to evolving by the time-independent Hamiltonian $H_G(s)\big|_{s=1/2}=\frac{1}{2}(I_2-\Delta\bar \sigma_x)$, which we denote $H_{\nicefrac{1}{2}}$. This Hamiltonian is the core of algorithms for the search problem: evolving by $H_{\nicefrac{1}{2}}$ would slowly rotate the system  to a state close to $\ket{y}$ \cite{Oshima01}.  Similarly, in the adiabatic algorithm \cite{RC02} the Hamiltonian spends most of the time close to the $H_{\nicefrac{1}{2}}$, where the gap is minimal, while the original gate model algorithm by Grover \cite{Grover96} can be seen as a simulation (or an approximation by Trotter formula \cite{NC00}) of the same Hamiltonian). 

We now discuss adding an additional modulation to algorithm $\mathcal A$ to improve its robustness while maintaining performance. We define algorithm $\mathcal B$ \footnote{See Supplemental Material at [URL will be inserted by publisher] for the spectrum}:
\begin{equation}
\begin{split}
H_{\mathcal{B}}(t)&=(I_N-\ket{u}\bra{u})\cdot \frac{1+A\cos(\omega t)}{2} \\
&+(I_N-\ket{y}\bra{y})\cdot\frac{1-A\cos(\omega t)}{2} - B\cos(\omega t) \ketbra{u}{u}
\\
H_{\mathcal{B}}\bigg|_V&=
\begin{bmatrix}
\frac{1}{2}-(B+\frac{A}{2})\cos(\omega t) &  -\frac{\Delta}{2} (B\cos(\omega t)+1)
\\
-\frac{\Delta}{2} (B\cos(\omega t)+1) & \frac{1}{2}+ \frac{A}{2}\cos(\omega t) 
\end{bmatrix}
\\
&+O(\Delta ^2).
\end{split}
\end{equation}
A natural question is whether Algorithm $\mathcal B$ is ``cheating'' by resources or by artificially increasing the gap. We use the opportunity for a small  discussion about resources. First, note that implementing $\ketbra{u}{u}$ requires no prior knowledge of $y$, namely the algorithm is the same for all $y$ (or $y$ is ``unknown''). This means that the total time duration $H_p$ is active would have to be  at least $2^{n/2}$ - otherwise it would contradict the optimality of Grover's algorithm \cite{BV97}. To understand the role of $B$, one can partition $H_{\mathcal B}$ by the Trotter approximation to slices of time independent Hamiltonians, where evolution by $H_p$ and by terms that are not $H_p$ alternate. In this picture increasing $\abs B$ corresponds to using a stronger quantum computer between calls to the black box, but has no effect on the query complexity of the problem (the total time $H_p$ is active).

In what follows, we compare the robustness (to control errors) of algorithm $\mathcal B$ versus applying a time-independent Hamiltonian $H_{\nicefrac{1}{2}}$, which corresponds to the standard gate model and adiabatic algorithms.

\begin{figure} 

    \centering

\ifcolor
\includegraphics [scale=0.625, trim=2.35cm 12.5cm 1cm 6.8cm,clip=true] {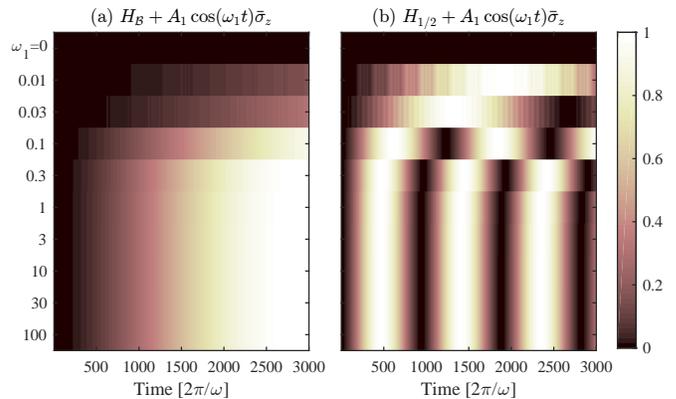}
\else
\includegraphics [scale=0.625, trim=2.35cm 12.5cm 1cm 6.8cm,clip=true] {NoiseInV2gray}
\fi
\caption{\label{fig:NoiseInV}{\ifcolor\else(Color online) \fi}16-qubit numerical simulation comparing the robustness of Algorithm $\mathcal B$ versus an evolution by $H_{\nicefrac{1}{2}}$ . Panel \textbf{(a)} correspond to Algorithm $\mathcal B$ with parameters $\omega=3.67,A=1,B=9.12$, and panel \textbf{(b)} corresponds to evolving by $H_{\nicefrac{1}{2}}$. The error $A_1 \cos (\omega_1 t) \bar \sigma_z$  with $A_1=0.05$ is equivalent to an error in $s(t)$. Each row in a panel is a simulation with different $\omega_1$ which is displayed on the y-axis. The brightness of the row changes from left to right as the value of $P_+$ varies in time under the noise of the specified $\omega_1$. Both algorithms are influenced by errors with $\omega_1\approx\Delta=0.125$, and fail as $\omega_1$ diminishes. However both are  generally robust to high frequency errors.} 

\end{figure} 

\noindent \emph{Hamiltonian control errors} are uncontrolled terms  causing the system to deviate unitarily from the intended evolution. The first error we focus on is in the form $A_1\cos(\omega_1 t+\varphi)\bar \sigma _z$ which preserves the subspace $V$ and represents an error in $s(t)$ (see Equation \ref{eq:HinV}).

Consider $H_{\nicefrac{1}{2}}$ with a harmonic control error in $s(t)$:
\begin{equation}
\widetilde H_{\nicefrac{1}{2}}=\frac{I_2}{2}-\frac{\Delta}{2}\bar \sigma_x + A_1\cos(\omega_1 t + \varphi) \bar \sigma_z.
\end{equation}
This is exactly the LZS Hamiltonian, therefore for high frequency errors ($\omega_1 \gg \Delta$) the Rabi frequency is $\widetilde \Omega=\Delta\abs{ J_0\prnt{\frac{A_1}{\omega_1}}}$, and the evolution is generally unaffected. On the other hand for $\omega_1=0$, even $A_1\approx \Delta$ may cause the system freezes in the initial state because the $\bar \sigma_z$ rotation  may become  more dominate than the desired $\bar \sigma_x$ rotation. Hence algorithms based on $H_{\nicefrac{1}{2}}$  are not robust to low frequency control errors. 

Algorithm $\mathcal {B}$ generally shows similar robustness (see Figure \ref{fig:NoiseInV}).  It fails to find $y$ when $\omega_1= 0$ and $A_1\approx \Delta$ for the same reasons $H_{\nicefrac{1}{2}}$ fails. For high frequency errors we write the Hamiltonian $H_{\mathcal{B}}+A_1\cos(\omega_1 t+\varphi )\bar \sigma_z$  in the appropriate rotating frame (around $\bar \sigma_z$)\footnote{See Supplementary Material [URL will be inserted by publisher]}: 
\begin{equation} \label{eq:AlgBNoiseInV}
\begin{split}
&\widetilde H'_{\mathcal{B}}\bigg|_V=
\begin{bmatrix}
0 &  -\frac{\Delta}{2} (B\cos(\omega t)+1)\chi
\\
~
\\
-\frac{\Delta}{2} (B\cos(\omega t)+1)\chi^* & 0
\end{bmatrix}
\\
&+O(\Delta ^2)
\\
&\chi = \sum_{k,k_1=-\infty}^\infty J_k\prnt{\frac{A+B}{\omega}} J_{k_1} \prnt{\frac{2A_1}{\omega_1}} e^{ik_1(\omega_1 t+\varphi) -ik\omega t}.
\end{split}
\end{equation}
 The algorithm is generally unaffected by high frequency errors ($\omega_1 \gg \Delta$) where all terms except $k=k_1=0$ average out, and the Rabi oscillation is $\widetilde \Omega=J_0\prnt{\frac{A+B}{\omega}}J_0\prnt{\frac{2A_1}{\omega_1}}$. Note that if for some $k,k_1$, $k_1\omega_1 \approx k\omega$, these terms would not average out may in principle cause the algorithm to fail because of CDT.

The second errors we consider in our comparison are errors that do not preserve  $V$. For their analysis, we use a three-level system toy model composed of the previously defined states $\ket{\bar 0},\ket{\bar 1}$ and an additional state $\ket{\bar{2}}$ which represents a state outside of $V$. The error term  we choose to focus on is the term $\eta(\ketbra{\bar 0}{\bar 2} +\ketbra{\bar 2}{\bar 0})$. The Hamiltonians take the form: 
\begin{equation} \label{eq:3ls}
\begin{split}
H_{\nicefrac{1}{2}}&=
\begin{bmatrix}
\frac{1}{2} &  -\frac{\Delta}{2}  & \eta 
\\
-\frac{\Delta}{2}  & \frac{1}{2} & 0 
\\
\eta & 0 & 1
\end{bmatrix}
\\
H_{\mathcal{B}}&=
\begin{bmatrix}
\frac{1}{2}-(B+\frac{A}{2})\cos(\omega t) &  -\frac{\Delta}{2} (B\cos(\omega t)+1) & \eta 
\\
-\frac{\Delta}{2} (B\cos(\omega t)+1) & \frac{1}{2}+ \frac{A}{2}\cos(\omega t) & 0 
\\
\eta & 0 & 1
\end{bmatrix}
\\
&+O(\Delta^2).
\end{split}
\end{equation}
Interestingly, $H_{\nicefrac{1}{2}}$ already have some inherent robustness to errors diverting the system to $\ket{\bar 2}$: the diagonal elements of $H_{\nicefrac{1}{2}}$ in equation \ref{eq:3ls} can be seen as ``potential energies'' of three sites. Therefore a particle in $\ket {\bar 0}$ needs to overcome a potential difference to reach $\ket{\bar 2}$, while it does not need to face a barrier when transitioning to $\ket{\bar 1}$.

Algorithm $\mathcal B$ improves the natural error suppression by adding CDT between the states $\ket{\bar 0}$ and $\ket{\bar 2}$, while allowing transitions between $\ket{\bar 0}$ and $\ket{\bar 1}$. We give here a simplified analysis using the rotating wave approximation, however we stress that finer tools such as Floquet theory \cite{Shirley65} better describe the dynamics of the system, and should be used when one attempts to find optimal values for $B,A,\omega$  (see Figure \ref{fig:3LS}). After changing to a rotating frame where $\bra{\bar 0}H_{\mathcal{B}}\ket{\bar 0}=\bra{\bar 1}H_{\mathcal{B}}\ket{\bar 1} =0$, and using the rotating wave approximation, we have:
\begin{equation} \label{eq:AlgB3LSNoise}
 H'_{\mathcal B} = \begin{bmatrix}
 0 &  -\frac{\Delta}{2} J_0\prnt{\frac{B+A}{\omega}} & \eta J_0\prnt{\frac{B+A/2}{\omega}}
\\
-\frac{\Delta}{2} J_0\prnt{\frac{B+A}{\omega}} & 0 & 0 
\\
\eta J_0\prnt{\frac{B+A/2}{\omega}} & 0 & \frac{1}{2}
\end{bmatrix}.
\end{equation}
By choosing $B,A, \omega$ s.t. $\frac{B+A/2}{\omega}$ is a root of $J_0$, the transition from $\ket{\bar{0}}$ to $\ket{\bar{2}}$ is suppressed. On the other hand the transition from $\ket{\bar 0}$ to $ \ket{\bar 1}$, which dominates $\Omega$ and the computation time, is only reduced by a factor of   $J_0\prnt{\frac{B+A}{\omega}}$. Figure \ref{fig:3LS} illustrate a scenario where Algorithm $\mathcal B$ is robust to a control error that ruins algorithms based on $H_{\nicefrac{1}{2}}$.

\begin{figure} 
\ifcolor
\includegraphics[scale=0.75, trim=4.2cm 8.7cm 6cm 9.5cm,clip=true]{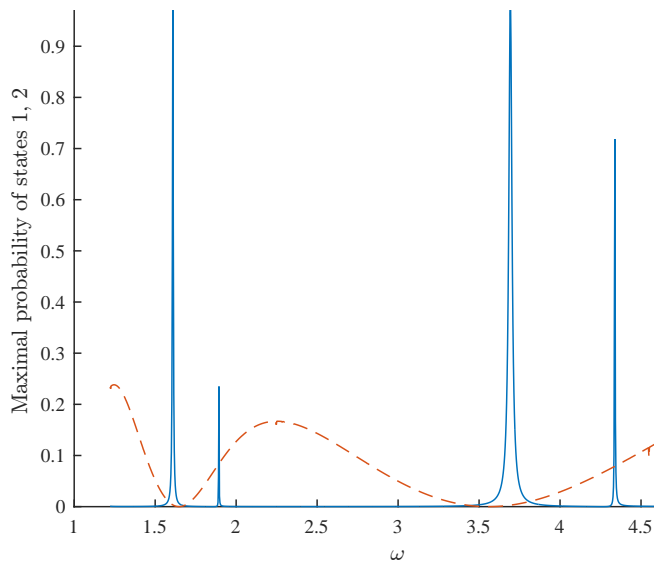}
\else
\includegraphics[scale=0.75, trim=4.2cm 8.7cm 6cm 9.5cm,clip=true]
{NoiseOutOfVdashedGray}
\fi
\caption{\label{fig:3LS}{\ifcolor\else(Color online)  \fi}A simulation of Algorithm $\mathcal B$ with control errors which do not preserve $V$. We set $n=20,A=1,B=9.12, \eta=0.3$, and simulated the three level system with different values of $\omega$ ($x$ axis). For every simulation,  two data points were plotted for the maximal probability reached by the states $\bar 1$ (blue, solid) and $\bar 2$ (orange, dashed) in the time interval $t=[0,150/\Delta]$. The ratio between the desired transition $\Delta/2\approx 5\cdot 10^{-4}$ and the control error $\eta$ is 1:600, and for algorithms based on $H_{\nicefrac{1}{2}}$ the maximal probability reached by the state $\bar 1$ is neglectable. The graph shows that for some $\omega$, the peak probability of  ${\bar 1}$ is close to one, hence Algorithm $\mathcal B$ is more robust to such errors. Note that equation \ref{eq:AlgB3LSNoise} predicts that  the transition $\bar 0 \rightarrow \bar 1$ peaks for $\omega\approx 1.74, 4$ which corresponds for the first two roots of $J_0\prnt{\frac{B+A/2}{\omega}}$, where $\bar 0\rightarrow\bar 2$ transition is strongly suppressed. The simulation shows that the transition  $\bar 0\rightarrow \bar 1$ peaks at \emph{two} frequencies around each root - this implies that the rotating wave approximation is insufficient to describe the dynamics of the system.}
\end{figure}

\emph{Thermal noise:} 
Implementing error correction for quantum algorithms based on continuous Hamiltonians is an open problem \cite{YSB13}. One can \emph{suppress} thermal noise (as well as control errors) by encoding the Hamiltonian by a stabilizer code \cite{Gottesman97}, combined with dynamical decoupling \cite{Lidar08}, energy gap protection \cite{JFS06}, or Zeno effect suppression \cite{PRDL12}; all of them function very similarly \cite{FLP04, YSB13}, providing enhanced performance for finite size systems, which were recently described in noisy intermediate scale quantum (NISQ)\cite{preskill18}. For exponential time algorithm such as the unstructured search problem, ultimately a logical error correction needs to be added.

\section{Discussion and conclusion}

In this Letter, we propose a new diabatic algorithm for solving the Grover problem using LZS interferometry. While the Grover problem is important on its own, it is interesting to examine the applicability of our paradigm to additional problems. It remains an open question whether one can translate any adiabatic algorithm to a diabatic algorithm. 

Diabaticity allowed us to suppress uncontrolled Hamiltonian terms using a mechanism inspired by coherent destruction of tunneling. We conjecture  the need for hybrid algorithms (diabatic/adiabatic), tailored to the noise parameters of a system.  

Finally it is interesting to find an expression for the optimal driving frequencies in Algorithm $\mathcal B$, their spectral width, and effectiveness.

\paragraph{Acknowledgments:}
\begin{acknowledgements}
The authors thank Michael Ben-Or, Dorit Aharonov, Tuvia Gefen, and Alex Retzker for the useful discussions. YA's work is supported by ERC grant number 280157,   and Simons foundation grant number 385590. YO's work is supported by ERC grant number 280157, and by ISF grant 1721/17. NK is supported by the ERC Project No. 335933.
\end{acknowledgements}


%

%

%
%
%

\newpage
\onecolumngrid
\section{Supplemental Material}

\subsection{Invariant  subspace in $H_G(s)$}

\begin{myclaim}
The subspace $V=\mathrm{span}\{\ket{y},\ket{u}\}$ is invariant to
\begin{equation} \label{eq:apndxHg} H_{\mathrm{G}}\left(s(t)\right)=(1-s(t))\cdot (I_N-\ket{u}\bra{u})+s(t)\cdot (I_N - \ketbra{y}{y}).
\end{equation}
\end{myclaim}
\begin{proof}
We that $H_{\mathrm G}(s(t))$ acting on any vector in $V$ keeps it in $V$ for all $s$. First, 
\begin{equation}
H_{\mathrm{G}}\ket{u} = s(t) \prnt{\ket{u} - \braket{u|y}\ket{y}} \in V
\end{equation}
\begin{equation}
H_{\mathrm{G}}\ket{y} = (1-s(t)) \prnt{\ket{y} - \braket{y|u}\ket{u}} \in V.
\end{equation}
A general vector in $V$ takes the form $\ket{v}=\prnt{\alpha \ket{u}+\beta\ket{y}} \in V$, and one can see that $H_{\mathrm G}(s(t)) \ket{v} \in V$ for any choice of $s,\alpha,\beta$. 
\end{proof}
This invariance allows us to reduce an $2^n$-dimensional problem to a 2-dimensional problem as required for the similarity relation in Claim \ref{clm:equiv}. Additionally this enables numerical simulations for high values of $n$.

\subsection{Proof of Claim \ref{clm:equiv}\label{sec:Equivalence}}
\setcounter{myclaim}{0}
Here we show the similarity of $H_{\mathrm G}$ in the subspace $V$ to the LZS Hamiltonian. It is clear that $\ket{u},\ket{y}$ are not orthogonal and therefore they cannot be mapped to $\ket{0},\ket{1}$ in $H_{\mathrm {LZS}}$. To overcome the problem we found a basis that is exponentially close to $\ket{u},\ket{y}$, which allows stating the similarity relation. Note that the rate $s(t)$ is also slightly adjusted. 

\begin{myclaim} 
The projection of $H_{\mathrm G}(s(t))$ on $V$ satisfies:

\begin{equation}
H_{\mathrm G}(s(t)) \bigg|_V = \left(\frac{I_2}{2} + \left(s(t)-\frac{1}{2}\right) \sqrt{1-\Delta^2} \bar\sigma_z -\frac{\Delta}{2}\bar\sigma_x \right)
\end{equation}
where $\Delta=\braket{y|u}$. The operators $\bar \sigma_x ,\bar \sigma_z$ act on the states
\begin{equation}
\begin{split}
\ket{\bar 0}&=\sqrt{\frac{1+\sqrt{1-\Delta^{2}}}{2}} \ket{u} + \sqrt{\frac{1-\sqrt{1-\Delta^{2}}}{2}}\ket{u^\perp} \\
\ket{\bar 1}&=\sqrt{\frac{1-\sqrt{1-\Delta^{2}}}{2}} \ket{u} - \sqrt{\frac{1+\sqrt{1-\Delta^{2}}}{2}}\ket{u^\perp}
\end{split}
\end{equation}
where  ${\ket{u^\perp}\mathrel{\mathop:}=
 \frac{\ket{y}-\Delta\ket{u}} {\sqrt{1-\Delta^2}}}$ is the vector orthogonal to $\ket{u}$ in $V$.
\end{myclaim}

\begin{proof}

As defined before,
\begin{equation} H_{\mathrm{G}}\left(s(t)\right)=(1-s(t))\cdot (I_N-\ket{u}\bra{u})+s(t)\cdot (I_N - \ketbra{y}{y}).
\end{equation}

We are to prove that the matrix form of $H_\mathrm G (s(t))$ projected on $V$, in the basis $\ket{\bar 0},\ket{\bar 1}$ is: 
\begin{equation} \label{eq:ApndxHGinV}
H_\mathrm G (s(t))\bigg|_V=\prnt{\begin{array}{cc}
\frac{1}{2}+\prnt{s(t)-\frac{1}{2}}\eta & -\frac{\Delta}{2} \\ 
-\frac{\Delta}{2} & \frac{1}{2}-\prnt{s(t)-\frac{1}{2}}\eta
\end{array}}
\end{equation}
where $\xi=\sqrt{1-\Delta^2}$. In other words we are to prove that 

\begin{equation}
\begin{split}
\bra{\bar 0} H_\mathrm{G}(s(t)) \ket{\bar 0}&=\frac{1}{2}+\prnt{s(t)-\frac{1}{2}}\xi
\\
\bra{\bar 1} H_\mathrm{G}(s(t)) \ket{\bar 1}&=\frac{1}{2}-\prnt{s(t)-\frac{1}{2}}\xi
\\
\bra{\bar 0} H_\mathrm{G}(s(t)) \ket{\bar 1}&=-\frac{\Delta}{2}
\end{split} 
\end{equation}

It is helpful to use the equalities in the calculation that follows: 
\begin{equation}
\begin{split}
\bra{u}H_\mathrm{G}(s(t)) \ket{u} &=  (1-\Delta^2)s(t)
\\
\bra{u^\perp}H_\mathrm{G}(s(t)) \ket{u} &=  \Delta\xi \cdot s(t)
\\
\bra{u^\perp}H_\mathrm{G}(s(t)) \ket{u^\perp} &= 1 - \xi^2 \cdot s(t).
\end{split}
\end{equation}

\begin{equation}
\begin{split}
\bra{\bar 0} H_\mathrm{G}(s(t)) \ket{\bar 0}&={\frac{1+\xi}{2}} \cdot (1-\Delta ^2) s(t) +   2 \Delta \xi s(t) \sqrt{\frac{1+\xi}{2}\cdot \frac{1-\xi}{2}} + (1-\xi^2 s(t)){\frac{1-\xi}{2}}
\\
&=\xi^2 {\frac{1+\xi}{2}}   s(t) +    \Delta^2 \xi s(t)  + (1-\xi^2 s(t)){\frac{1-\xi}{2}} 
\\
&=
s(t)\cdot \prnt{\xi^2 {\frac{1+\xi}{2}} + (1-\xi^2)\xi - \xi^2 \frac{1-\xi}{2}} + \frac{1-\xi}{2}=\frac{1}{2}+\prnt{s(t)-\frac{1}{2}}\xi
\end{split}
\end{equation}

\begin{equation}
\begin{split}
\bra{\bar 1} H_\mathrm{G}(s(t)) \ket{\bar 1}
&=
{\frac{1-\xi}{2}} \cdot (1-\Delta ^2) s(t) -  2 \Delta \xi s(t) \sqrt{\frac{1+\xi}{2}\cdot \frac{1-\xi}{2}} + (1-\xi^2 s(t)){\frac{1+\xi}{2}}
\\
&=
\xi^2 {\frac{1-\xi}{2}}   s(t) -    \Delta^2 \xi s(t)  + (1-\xi^2 s(t)){\frac{1+\xi}{2}} 
\\
&=
s(t)\prnt{\xi^2 \frac{1-\xi}{2}-(1-\xi^2)\xi - \xi^2\frac{1+\xi}{2}}+\frac{1+\xi}{2} = \frac{1}{2}-\prnt{s(t)-\frac{1}{2}}\xi
\end{split}
\end{equation}

\begin{equation}
\begin{split}
\bra{\bar 1} H_\mathrm{G}(s(t)) \ket{\bar 0}
&=(1-\Delta^2)s(t)\sqrt{\frac{(1+\xi)(1-\xi)}{4}}+  \Delta\xi \cdot s(t) \prnt{\frac{1-\xi}{2}-\frac{1+\xi}{2}} - (1-\xi^2\cdot s(t))\sqrt{\frac{(1-\xi)(1+\xi)}{4}}
\\
&= \frac{\xi^2 \Delta}{2} s(t) - \Delta \xi^2 s(t) - (1-\xi^2 \cdot s(t))\frac{\Delta}{2}= - \frac{\Delta}{2}  
\end{split}
\end{equation}

We found all the elements of $H_{\mathrm{G}}$ in $V$, and proved equation \ref{eq:ApndxHGinV} is correct. The proof of Claim \ref{clm:equiv} follows.

\end{proof}

\subsection{Analysis of LZS oscillations using the rotating wave approximation.}

In this section we analyze the LZS oscillations and the robustness to errors by generalizing the rotating wave approximation analysis by \cite{Pegg73,AJZN07}.
\setcounter{myclaim}{2}
\begin{myclaim} \label{clm:RWA}
The Rabi frequency of  a system driven by $H_\mathrm{LZS}$ is 
\begin{equation} 
\Omega = \Delta \left|J_0 \left(\frac{A}{\omega}\right)\right|
\end{equation}
\end{myclaim}

\begin{proof}
We start with 2-level system and a general control function:
\begin{equation}
H(t) := \frac{1}{2} \left(
-a(t)\sigma_z -\Delta \sigma_x  
\right)=\frac{1}{2}\begin{bmatrix}
	-a(t) & -\Delta \\
	-\Delta & a(t)
\end{bmatrix}.
\end{equation}
Changing to  the rotating frame yields 

\begin{equation}
\ket{\psi(t)} = U(t) \ket{\psi'(t)},
\end{equation}
where
\begin{equation}
U(t)= \exp \prnttt{\frac{i}{2} \sigma_z \int a(t) dt}
\end{equation}
Note that the populations of the ground state the excited states are invariant to this transformation. The effective Hamiltonian $H'$ which satisfies the Shr\"odinger equation in the rotating frame, i.e.,
\begin{equation}
i\frac{d}{dt}\ket{\psi'} = H'\ket{\psi'}
\end{equation}
is the following \cite{Joachain1975}:
\begin{equation} \label{eq:rotatingNoiseV}
H'(t)=U^\dagger(t) H(t) U(t) - i U^\dagger(t)\frac{dU(t)}{dt} 
=
-\frac{\Delta}{2} \begin{bmatrix}
	0 & e^{-i\int a(t)dt} \\
	e^{i\int a(t)dt} & 0
\end{bmatrix}
\end{equation}
Assigning $a(t)=A\cos(\omega t)$, integrating and using the identity  
\begin{equation} \label{eq:BesselIdentity}
\exp\prnttt{iz\sin \gamma} = \sum_{k=-\infty}^\infty J_k(z)e^{ik\gamma}
\end{equation}
we get
\begin{equation}
H'(t)=
-\frac{\Delta}{2}
\begin{bmatrix}
	0 & \exp\prnttt{-i \frac{A}{\omega} \sin(\omega t)} \\
	\exp\prnttt{i \frac{A}{\omega} \sin(\omega t)} & 0
\end{bmatrix}
=
-\frac{\Delta}{2}
\begin{bmatrix}
	0 & \sum_{k=-\infty}^\infty J_k(A/\omega)e^{-ik\omega t} \\
	\sum_{k=-\infty}^\infty J_k(A/\omega)e^{ik\omega t} & 0
\end{bmatrix}
\end{equation}
Since $\omega \gg \Delta$ we can use the rotating wave approximation
\begin{equation}
H'(t)
=
-\frac{\Delta}{2}
\begin{bmatrix}
	0 &  J_0(A/\omega) \\
	 J_0(A/\omega) & 0
\end{bmatrix}
\end{equation}
and the proof follows.

\end{proof}

\/*
\yanote{here should be analysis Algorithm A/B with noises in V. It is a variance of the following claim}
\begin{myclaim}
The LZS oscillations are robust to harmonic noise  in the control function, with frequency  $\omega_1\ll \omega$ and amplitude $A_1<\omega_1, A$.
\end{myclaim}
\begin{proof}
We assume an harmonic noise:
\begin{equation}
a(t)=A\cos(\omega t) + A_1 \cos (\omega_1 t)
\end{equation}

The Hamiltonian $H'$ takes the form 

\begin{equation}
\begin{split}
H'(t)&=
-\frac{\Delta}{2}
\begin{bmatrix}
	0 & \exp\prnttt{-i \prnt{\frac{A}{\omega} \sin(\omega t)+\frac{A_1}{\omega_1} \sin(\omega_1 t)}} \\
	\exp\prnttt{+i \prnt{\frac{A}{\omega} \sin(\omega t)+\frac{A_1}{\omega_1} \sin(\omega_1 t)}} & 0
\end{bmatrix}
\\
&=
-\frac{\Delta}{2}
\begin{bmatrix}
	0 & \sum_{k,k_1=-\infty}^\infty J_k(A/\omega)J_{k_1}(A_1/\omega_1)e^{-it(k\omega +k_1\omega_1) } \\
	\sum_{k,k_1=-\infty}^\infty J_k(A/\omega)J_{k_1}(A_1/\omega_1)e^{it(k\omega +k_1\omega_1) } & 0
\end{bmatrix}
\end{split}
\end{equation}

We see that the time independent component corresponding to $k,k_1=0$  causes the system oscillates between $\ket{0}$ and $\ket{1}$ with frequency
$\Omega = {\Delta} \abs{J_0 \prnt{\frac{A}{\omega}}J_0 \prnt{\frac{A_1}{\omega_1}}}$.  Since $A_1 < \omega_1$ we are safe from $J_0$ zeros which would nullify $\Omega$, as well as the decay of $J_0(z)$ for large values of $z$.

Repeating the analysis of the proof in claim \ref{clm:RWA}, we see that $V_I$ in this case contains terms of order $O\prnt{J_k \prnt{\frac{A}{\omega}}J_{k_1}\prnt{\frac{A_1}{\omega_1}}\cdot \frac{\Delta}{k\omega-k_1\omega_1}}$. The assumption $\omega\gg \omega_1$  implies $\abs{k_1}\gg 1$, and since $J_k(z)\approx \frac{z^k}{2^k k!}$ for $z\ll 1$, these terms decay quickly. \yanote{mmm} The proof follows.

\end{proof}
\*/
%
%
%
%

\/*
\subsection{errors in $\bar \sigma_x, \bar \sigma_z$}
\yanote{for myself}

 Consider the Hamiltonian 
\begin{equation}
\begin{split}
H(t) &= \frac{1}{2} \prnt{-\Delta \sigma_x -A\cos(\omega t)\sigma_z - A_1 \cos(\omega_1 t) \vec{r}.\vec{\sigma}}
\end{split}
\end{equation} 
%

Following \cite{AJZN07} we switch to the rotating frame with the transformation $\ket{\psi(t)} = U(t) \ket{\psi'(t)}$, where
\begin{equation} 
U(t)= \exp \prnttt{\frac{i}{2} \prnt{\frac{A\sin(\omega t)}{\omega} + \frac{A_1\sin(\omega_1 t)}{\omega_1}} \sigma_z }.
\end{equation} 
The effective Hamiltonian in the rotating frame takes the form
\begin{equation}
\begin{split}
&H'(t)=U^\dagger(t) H(t) U(t) - i U^\dagger(t)\frac{dU(t)}{dt} 
\\
&=
\sigma_- e^{i\prnt{\frac{A\sin(\omega t)}{\omega} + \frac{A_1\sin(\omega_1 t)}{\omega_1}}} \prnt{-\frac{\Delta}{2} - A_1 r_{xy} \cos(\omega_1 t) e^{-i\varphi}}
\\
&+ h.c.
\end{split}
\end{equation}
where $\varphi= \mathrm{atan}(r_y/r_x)$,  $r_{xy}^2=r_x^2 + r_y^2$ and $\sigma_- = \sigma_x + i \sigma_y$.
%
By using the identity 
\begin{equation}
\exp\prnttt{iz\sin \gamma} = \sum_{k=-\infty}^\infty J_k(z)e^{ik\gamma}
\end{equation}
we get

\begin{equation} 
\begin{split}
&H'(t)=
\\
&\sum_{k,k_1=-\infty}^\infty \hspace{-10pt} J_k\prnt{\frac{A}{\omega}}J_{k_1}\prnt{\frac{A_1}{\omega_1}} {\biggl [}-\frac{\Delta}{2}\prnt{\cos(\xi t)\sigma_x +\sin(\xi t) \sigma_y }
\\
&+ A_1 r_{xy} \cos(\omega_1 t) \prnt{\cos(\xi t+\varphi)\sigma_x +\sin(\xi t+\varphi) \sigma_y } {\biggl ]}
\end{split}
\end{equation}


\yanote{what happens here? and is it too messy to conclude anything? perhaps i just want to handle sigmaz errors?}
where $\xi=k\omega+k_1\omega_1$.
*/

\subsection{Analysis of Algorithm $\mathcal{B}$ using the rotating wave approximation}

We give here a more detailed derivation of some of the rotating frame transformation of $H_{\mathcal B}$ in the main text (following  \cite{AJZN07}). In the case of $\bar \sigma_z$ error,
\begin{myclaim}
Let
\begin{equation}
\widetilde H_{\mathcal B}\bigg|_V  =
\begin{bmatrix}
\frac{1}{2}-(B+\frac{A}{2})\cos(\omega t) &  -\frac{\Delta}{2} (B\cos(\omega t)+1)
\\
-\frac{\Delta}{2} (B\cos(\omega t)+1) & \frac{1}{2}+ \frac{A}{2}\cos(\omega t) 
\end{bmatrix}+ A_1\bar \sigma_z\cos(\omega_1 t+\varphi) + O(\Delta ^2).
\end{equation}
Using a rotation around $\bar \sigma_z$ the effective Hamiltonian is as Equation \ref{eq:AlgBNoiseInV}:
\begin{equation}
\begin{split}
\widetilde H'_{\mathcal{B}}\bigg|_V&=
\begin{bmatrix}
0 &  -\frac{\Delta}{2} (B\cos(\omega t)+1)\chi
\\
~
\\
-\frac{\Delta}{2} (B\cos(\omega t)+1)\chi^* & 0
\end{bmatrix}
+O(\Delta ^2)
\\
\chi &= \sum_{k,k_1=-\infty}^\infty J_k\prnt{\frac{A+B}{\omega}} J_{k_1} \prnt{\frac{2A_1}{\omega_1}} e^{ik_1(\omega_1 t+\varphi) -ik\omega t}
\end{split}
\end{equation}
\end{myclaim}

\begin{proof}

First the global (time dependent) energy offset $\frac{1}{2}+B/2\cos(\omega t)$ is removed.  $O(\Delta^2)$ can be neglected since the Hamiltonian is applied for duration $O(1/\Delta)$. We get

\begin{equation}
\widetilde H_{\mathcal B}\bigg|_V  = \prnt{-\frac{B+A}{2}\cos(\omega t) + A_1 \cos(\omega_1 t +\varphi)} \bar \sigma_z -\frac{\Delta}{2}(B\cos(\omega t)+1)\bar \sigma_x
\end{equation}

Next we choose a rotating frame where the diagonal is zero in a similar way to equation \ref{eq:rotatingNoiseV}, but with $a(t)=(A+B)\cos(\omega t) - 2A_1 \cos (\omega_1 t+\varphi)$:

\begin{equation}
\begin{split}
\widetilde H_{\mathcal B}'\bigg|_V  &=
\begin{bmatrix}
0 &  -\frac{\Delta}{2} (B\cos(\omega t)+1) e^{-i\int a(t)dt}
\\
-\frac{\Delta}{2} (B\cos(\omega t)+1) e^{i\int a(t)dt} & 0
\end{bmatrix}
\\
&= -\frac{\bar \sigma_+}{2} \cdot  {\frac{\Delta}{2} (B\cos(\omega t)+1)} \exp\prnttt{-i \prnt{\frac{A+B}{\omega} \sin(\omega t)-\frac{2A_1}{\omega_1} \sin(\omega_1 t +\varphi)}} + h.c.
\\
&= -\frac{\bar \sigma_+}{2} \cdot {\frac{\Delta}{2} (B\cos(\omega t)+1)} \sum_{k,k_1=-\infty}^\infty J_k\prnt{\frac{A+B}{\omega}}e^{-ik\omega t} \cdot J_{k_1} \prnt{\frac{2A_1}{\omega_1}} e^{ik_1(\omega_1 t+\varphi)} + h.c.,
\end{split}
\end{equation}
where $\bar\sigma_+=\bar\sigma_x+\bar\sigma_y$. 
\end{proof}

Similarly we derive the transformation of the three level system in equation \ref{eq:AlgB3LSNoise}. 

\begin{myclaim}
Let 
\begin{equation}
H_{\mathcal{B}}=
\begin{bmatrix}
\frac{1}{2}-(B+\frac{A}{2})\cos(\omega t) &  -\frac{\Delta}{2} (B\cos(\omega t)+1) & \eta 
\\
-\frac{\Delta}{2} (B\cos(\omega t)+1) & \frac{1}{2}+ \frac{A}{2}\cos(\omega t) & 0 
\\
\eta & 0 & 1
\end{bmatrix}
\end{equation}
By rotating around $\bar \sigma_z$ the Hamiltonian can be approximated by equation \ref{eq:AlgB3LSNoise}:
\begin{equation} 
 H'_{\mathcal B} = \begin{bmatrix}
 0 &  -\frac{\Delta}{2} J_0\prnt{\frac{B+A}{\omega}} & \eta J_0\prnt{\frac{B+A/2}{\omega}}
\\
-\frac{\Delta}{2} J_0\prnt{\frac{B+A}{\omega}} & 0 & 0 
\\
\eta J_0\prnt{\frac{B+A/2}{\omega}} & 0 & \frac{1}{2}
\end{bmatrix}.
\end{equation}
\end{myclaim}
\begin{proof}
Initially we  change the  reference frame by the first equality of equation \ref{eq:rotatingNoiseV}, with 
\begin{equation}
U=\exp\prnttt{{i}\prnt{B+\frac{A}{2}}\frac{\sin(\omega t)}{\omega}\ketbra{\bar 0}{\bar 0} - i\frac{A \sin(\omega t)}{2\omega} \ketbra{\bar 1}{\bar 1}},
\end{equation}
we get:
\begin{equation}
H_{\mathcal{B}}'=
\begin{bmatrix}
\frac{1}{2} &  -\frac{\Delta}{2} (B\cos(\omega t)+1) e^{-i\frac{A+B}{\omega}\sin(\omega t)} & \eta e^{-i\frac{B+A/2}{\omega}\sin(\omega t)}
\\
-\frac{\Delta}{2} (B\cos(\omega t)+1)  e^{i\frac{A+B}{\omega}\sin(\omega t)}  & \frac{1}{2} & 0 
\\
\eta  e^{i\frac{B+A/2}{\omega}\sin(\omega t)} & 0 & 1
\end{bmatrix}
\end{equation}
The diagonal can be adjusted by subtracting $\frac{1}{2}I$. The proof is concluded by using the Bessel identity in equation \ref{eq:BesselIdentity}, and by neglecting all but the zero frequency terms (rotating wave approximation).

\end{proof}

\begin{figure} [h]
\includegraphics[scale=0.65, trim=2.2cm 10cm 1cm 10cm,clip=true]{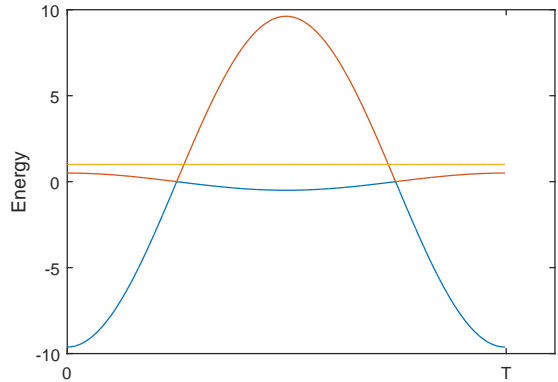}
\caption{The spectrum of the noiseless $H_{\mathcal{B}}$ over one period. The parameters are $n=16, A=1, B=9.1193$. Note that the yellow energy level is outside the invariant subspace $V$.}
\end{figure}

\subsection{Units consistency}
In our analysis we have generally ignore units (e.g., energy, frequency), specially because computational/query complexity is invariant to multiplicative factors. Here we rewrite the main results while keeping the units consistent. The problem Hamiltonian is normally given with an energy scale $\varepsilon$  ($\hbar=1$):
\begin{equation}
H_{p}=\varepsilon(I_N-\ket{y}\bra{y})
\end{equation}
The Hamiltonian evolution in Algorithm $\mathcal A$ is the following:
\begin{equation}
\begin{split}
H_{G}\left(s(t)\right)&=\varepsilon \prntt{
 {(1-s(t))\cdot (I_N-\ket{u}\bra{u})+ s(t)\cdot(I_N-\ket{y}\bra{y})}}
 \\
 s(t)&=\frac{1-a\cos(\omega t)}{2},
\end{split}
\end{equation}
where $a\in [0,1]$ is the dimensionless amplitude of the control function $s(t)$. Note that the minimal energy gap is $\Delta=2^{-n/2} \varepsilon$.

Claim \ref{clm:equiv} takes the following form:
\setcounter{myclaim}{0}
\begin{myclaim}
The projection of $H_{\mathrm G}(s(t))$ on $V$ satisfies:
\begin{equation} 
\begin{split}
H_{\mathrm G}(s(t)) \bigg|_V &= \varepsilon\prntt{\frac{I_2}{2} + \left(s(t)-\frac{1}{2}\right) \sqrt{1-\delta^2} \bar\sigma_z -\frac{\delta}{2}\bar\sigma_x }
\\
&=\varepsilon\prntt{\frac{I_2}{2} + \left(\frac{-a\cos(\omega t)}{2}\right) \sqrt{1-\delta^2} \bar\sigma_z -\frac{\delta}{2}\bar\sigma_x }
\end{split}
\end{equation}
where $\delta=\braket{y|u}$ is dimensionless. The operators $\bar \sigma_x ,\bar \sigma_z$ act on the states
\begin{equation}
\begin{split}
\ket{\bar 0}&=\sqrt{\frac{1+\sqrt{1-\delta^{2}}}{2}} \ket{u} + \sqrt{\frac{1-\sqrt{1-\delta^{2}}}{2}}\ket{u^\perp} \\
\ket{\bar 1}&=\sqrt{\frac{1-\sqrt{1-\delta^{2}}}{2}} \ket{u} - \sqrt{\frac{1+\sqrt{1-\delta^{2}}}{2}}\ket{u^\perp}
\end{split}
\end{equation}
where  ${\ket{u^\perp}\mathrel{\mathop:}=
 \frac{\ket{y}-\delta\ket{u}} {\sqrt{1-\delta^2}}}$ is the vector orthogonal to $\ket{u}$ in $V$.
\end{myclaim}

The rotating frame approximation holds when $\omega \gg \Delta=2^{-n/2}\varepsilon$. The run time of the algorithm in this case is inverse proportional to the Rabi frequency $\Omega=\varepsilon \delta \cdot J_0\prnt{\frac{\sqrt{1-\delta^2}a\varepsilon}{\omega}}$. On the other hand, when $ \varepsilon  a \omega \ll \Delta^2$, the process is adiabatic. 

Algorithm $\mathcal B$ is defined using an additional dimensionless variable $b$:
\begin{equation}
\begin{split}
H_{\mathcal{B}}(t)&=\varepsilon \prnt{(I_N-\ket{u}\bra{u})\cdot \frac{1+a\cos(\omega t)}{2} 
+(I_N-\ket{y}\bra{y})\cdot\frac{1-a\cos(\omega t)}{2} - b\cos(\omega t) \ketbra{u}{u}}
\\
H_{\mathcal{B}}\bigg|_V&=\varepsilon
\begin{bmatrix}
\frac{1}{2}-(b+\frac{a}{2})\cos(\omega t) &  -\frac{\delta}{2} (b\cos(\omega t)+1)
\\
-\frac{\delta}{2} (b\cos(\omega t)+1) & \frac{1}{2}+ \frac{a}{2}\cos(\omega t) 
\end{bmatrix}+\varepsilon(\abs a+\abs b )\cdot O(\delta ^2),
\end{split}
\end{equation}

and by adding a unitary error from $V$ to $V^\perp$ it takes the form:
\begin{equation} 
\begin{split}
&H_{\mathcal{B}}=
\varepsilon
\begin{bmatrix}
\frac{1}{2}-(b+\frac{a}{2})\cos(\omega t) &  -\frac{\delta}{2} (b\cos(\omega t)+1) & \eta 
\\
-\frac{\delta}{2} (b\cos(\omega t)+1) & \frac{1}{2}+ \frac{a}{2}\cos(\omega t) & 0 
\\
\eta & 0 & 1
\end{bmatrix}+\varepsilon ( \abs a + \abs b) \cdot O(\delta^2).
\end{split}
\end{equation}
Finally, the optimal values for $a,b,\omega$ are in proximity to the roots of $J_0\prnt{\frac{\varepsilon(b+a/2)}{\omega}}$.

\end{document}